\documentclass[journal]{IEEEtran}

\ifCLASSINFOpdf
\else
\fi

\usepackage{mathtools}
\usepackage{amsmath}
\usepackage{graphicx}
\usepackage{amssymb}
\usepackage{amsthm}
\usepackage{cite}
\usepackage{mathtools}
\usepackage{bbold}
\usepackage{steinmetz}
\usepackage{stmaryrd}
\usepackage{bbm, dsfont}
\usepackage{verbatim}
\usepackage[latin1]{inputenc}
\usepackage{tikz}
\usetikzlibrary{shapes,arrows}
\usepackage[caption=false,font=footnotesize]{subfig}
\usepackage{enumerate}
\usepackage{color}
\usepackage{subfig}
\usepackage{slashbox}
\usepackage{array}
\usepackage{makecell}

\usepackage{pgfplots}
\usetikzlibrary{calc}

\usepackage{lipsum}
\usepackage[variablett]{lmodern}
\usepackage{amsbsy}
\usepackage{bm}
\usetikzlibrary{positioning}
\usetikzlibrary{shapes,arrows}

\pgfplotsset{
	every tick label/.append style={scale=1},
	every axis/.append style={
	}
}

\usepackage[]{algorithm}
\usepackage[noend]{algpseudocode}

\makeatletter
\def\BState{\State\hskip-\ALG@thistlm}
\makeatother

\newcommand{\rlp}[1]{\left(#1\right)}
\newcommand{\rlbo}[1]{\lefto[#1\right]}
\newcommand{\lefto}{\mathopen{}\left}
\newcommand{\kin}{K_{\mathrm{in}}}
\newcommand{\nb}{n_b}
\newcommand{\nf}{n_f}
\newcommand{\kout}{K_{\mathrm{out}}}

\newcommand{\niop}{\mathbf{n}_{\mathrm{idle}}}
\newcommand{\nbop}{\mathbf{n}_{\mathrm{busy}}}
\newcommand{\rinter}{R_\mathrm{inter}}
\newcommand{\setx}[2]{\mathcal{S}_{#1,#2}}
\newcommand{\lnk}[2]{L_{#1,#2}}
\newcommand{\wls}[2]{W_{#1,#2}}
\newcommand{\setxt}[2]{\tilde{\mathcal{S}}_{#1,#2}}
\newcommand{\pr}{\mathrm{Pr}}
\newcommand{\expe}{\mathrm{E}}
\newcommand{\bp}[1]{b_{#1}}
\newcommand{\bpl}[1]{\tilde b_{#1}}
\newcommand{\BR}{\mathrm{BP}}
\newcommand{\mi}{m_{\mathrm{1}}}
\newcommand{\mj}{m_{\mathrm{2}}}
\newcommand{\mk}{m_{\mathrm{3}}}
\newcommand{\ml}{m_{\mathrm{4}}}
\newcommand{\mm}{m_{\mathrm{5}}}

\newcommand{\nidt}{\mathcal{T}_{\text{inter}}}

\newtheorem{lemma}{Lemma}

\begin{document}

\title{ Exploiting AWG Free Spectral Range Periodicity in Distributed Multicast Architectures }

\author{Kamran~Keykhosravi,\ 
	 Houman Rastegarfar,\ 
	Nasser Peyghambarian,\
	and
	Erik~Agrell
 
	\thanks{K. Keykhosravi and E. Agrell are with the Department of Electrical Engineering, Chalmers University of Technology, Gothenburg 412 96, Sweden (e-mail:
		kamrank@chalmers.se; agrell@chalmers.se).}
		\thanks{H. Rastegarfar and N. Peyghambarian are with the College of Optical Sciences, University of Arizona, Tucson, Arizona 85721, USA (e-mail:
		houman@optics.arizona.edu; nasser@optics.arizona.edu).}
	\thanks{Manuscript received August 14, 2018; revised August 14, 2018.}
	
}

\maketitle

\begin{abstract}
	Modular optical switch architectures combining wavelength routing based on arrayed waveguide grating (AWG) devices and multicasting based on star couplers hold promise for flexibly addressing the exponentially growing traffic demands in a cost- and power-efficient fashion. In a default switching scenario, an input port of the AWG is connected to an output port via a single wavelength.  This can severely limit the capacity between broadcast domains, resulting in interdomain traffic switching bottlenecks. In this paper, we examine the possibility of resolving capacity bottlenecks by exploiting multiple AWG free spectral ranges (FSRs), i.e., setting up multiple parallel connections between each pair of broadcast domains. To this end, we introduce a multi-FSR scheduling algorithm for interconnecting broadcast domains by fairly distributing the wavelength resources among them. We develop a general-purpose analytical framework to study the blocking probabilities in a multistage switching scenario and compare our results with Monte Carlo simulations. Our study points to significant improvements with a moderate increase in the number of FSRs. We show that an FSR count beyond four results in diminishing returns. Furthermore, to investigate the trade-offs between the network- and physical-layer effects, we conduct a cross-layer analysis, taking into account pulse amplitude modulation (PAM) and rate-adaptive forward error correction (FEC). We illustrate how the effective bit rate per port increases with an increase in the number of FSRs. 
	 
\end{abstract}


\IEEEpeerreviewmaketitle

\section{Introduction}


With the proliferation of smart mobile devices, the continuous advances in computational power, and the breakthroughs in the field of machine learning, the fifth generation of cellular networks (5G) is being rolled out to provide dramatic improvements in the throughput, latency, and reliability performance for a myriad of services and applications \cite{agyapong2014design, wang2015backhauling, velasco2017meeting, raza2016demonstration, gowda2016quasi}. With new technologies such as the Internet of Things (IoT), high-resolution video streaming, road safety, wearable devices, and augmented reality, and with increasing capacity demands for large-scale scientific calculations, the network traffic is growing at an exponential pace across all geographical spans. To cope with the ever-increasing traffic rates in a sustainable fashion, innovative and intelligent networking solutions that simultaneously optimize the transmission, architecture, and control and management aspects have become indispensable \cite{ tomkos2013optical, JiajiaChen_15, rastegarfar2017PAM}.

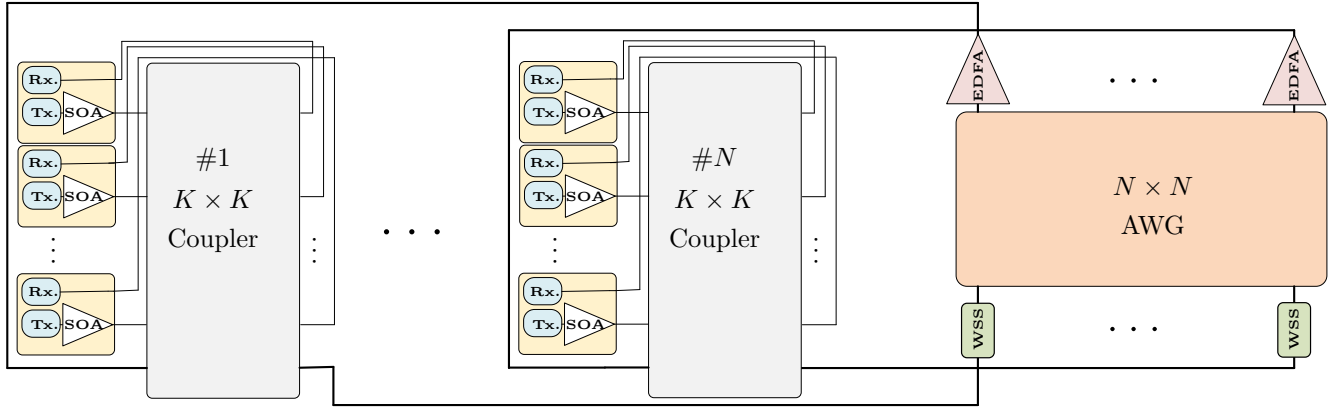
\begin{figure*}[!t]
	\centering
	\begin{tikzpicture}
	\node (ul) at (0,0) {};
	\node (cpp) [right=0cm of ul]{ 	\includegraphics[width=\linewidth]{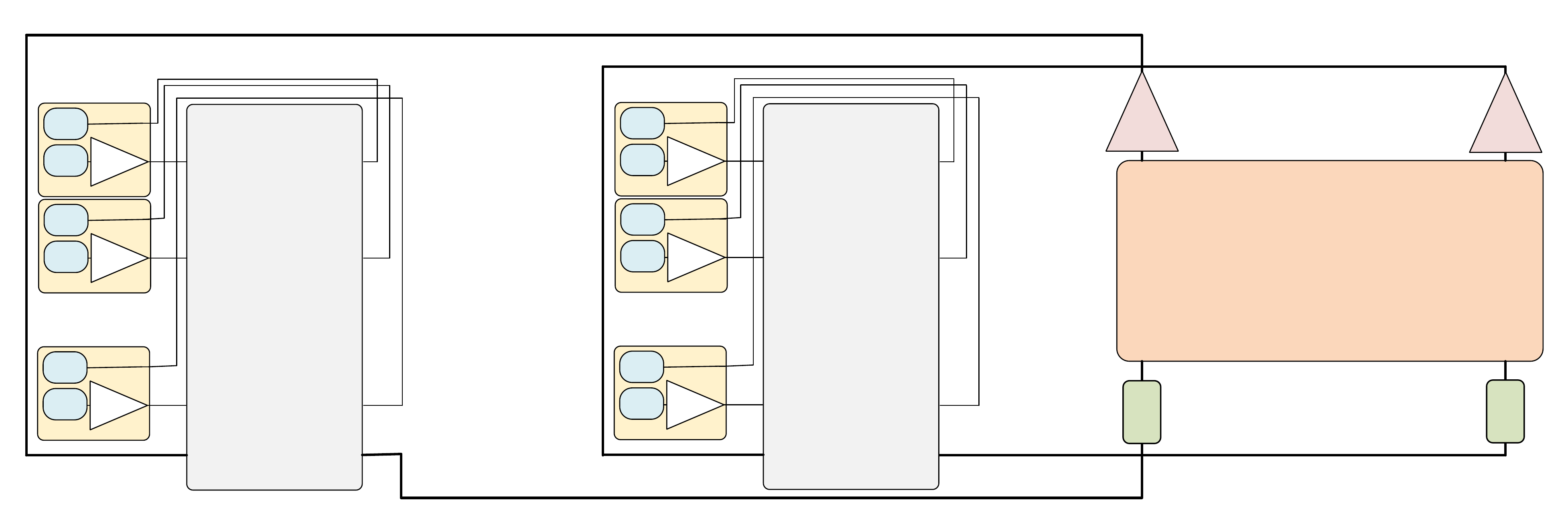}};
	\node (cp)[ right=2.5 of ul]{$K\times K$};
	\node (cp2)[ below=0 of cp]{Coupler};
	\node (n1)[ above=0 of cp]{$\#1$};
	\node (T11)[  below left=-1.61cm and 1.26cm  of cp]{\tiny \textbf{Tx.}};
	\node (R11)[  below left=-2.05cm and 1.26cm  of cp]{\tiny \textbf{Rx.}};
	\node (SOA11)[  right  =-.18cm of T11]{\tiny \textbf{SOA}};
	\node (T12)[  below left=-.5cm and 1.26cm  of cp]{\tiny \textbf{Tx.}};
	\node (R12)[  below left=-0.95cm and 1.26cm  of cp]{\tiny \textbf{Rx.}};
	\node (SOA12)[  right  =-.18cm of T12]{\tiny \textbf{SOA}};
	\node (T13)[  below left=1.2cm and 1.26cm  of cp]{\tiny \textbf{Tx.}};
	\node (R13)[  below left=.77cm and 1.26cm  of cp]{\tiny \textbf{Rx.}};
	\node (SOA13)[  right  =-.18cm of T13]{\tiny \textbf{SOA}};
	\node (cp2)[ right=9.17 of ul]{$K\times K$};
	\node (cp22)[ below=0 of cp2]{Coupler};
	\node (n2)[ above=0 of cp2]{$\#N$};
	\node (T21)[  below left=-1.61cm and 1.26cm  of cp2]{\tiny \textbf{Tx.}};
	\node (R21)[  below left=-2.05cm and 1.26cm  of cp2]{\tiny \textbf{Rx.}};
	\node (SOA21)[  right  =-.18cm of T21]{\tiny \textbf{SOA}};
	\node (T22)[  below left=-.5cm and 1.26cm  of cp2]{\tiny \textbf{Tx.}};
	\node (R22)[  below left=-0.95cm and 1.26cm  of cp2]{\tiny \textbf{Rx.}};
	\node (SOA22)[  right  =-.18cm of T22]{\tiny \textbf{SOA}};
	\node (T23)[  below left=1.2cm and 1.26cm  of cp2]{\tiny \textbf{Tx.}};
	\node (R23)[  below left=.77cm and 1.26cm  of cp2]{\tiny \textbf{Rx.}};
	\node (SOA23)[  right  =-.18cm of T23]{\tiny \textbf{SOA}};
	\node (EDFA1)[ rotate=90, below left=-2.4cm and -10.65cm  of cp]{\tiny \textbf{EDFA}};
	\node (EDFA2)[ rotate=90, below left=-2.4cm and -14.85cm  of cp]{\tiny \textbf{EDFA}};
	\node (WSS1)[ rotate=90, below left=1.05cm and -10.65cm  of cp]{\tiny \textbf{WSS}};
	\node (WSS2)[ rotate=90, below left=1.05cm and -14.85cm  of cp]{\tiny \textbf{WSS}};
	\node (AWG)[ above right=-0.2cm and 15 of ul]{$N\times N$};
	\node (AWG2)[ below=0 of AWG]{AWG};
	\node (dotsCoupler11)[ rotate=90, below left=0.03cm and 1.6cm  of cp]{\dots};
	\node (dotsCoupler12)[ rotate=90, below left=0.03cm and -1.87cm  of cp]{\dots};
	\node (dotsCoupler21)[ rotate=90, below left=0.03cm and 1.6cm  of cp2]{\dots};
	\node (dotsCoupler22)[ rotate=90, below left=0.03cm and -1.87cm  of cp2]{\dots};
	\node (dotsCoupler12)[  below left=0cm and -4cm  of cp]{\Huge\dots};
	\node (dotsAWG1)[  above left=1cm and -1cm  of AWG]{\huge\dots};
	\node (dotsAWG2)[  below=3cm  of dotsAWG1]{\huge\dots};
	\end{tikzpicture}
	
	\caption{\footnotesize A distributed multicast architecture based on star couplers and AWG \cite{rastegarfar2017PAM}. SOA: semiconductor optical amplifier. EDFA: erbium-doped fiber amplifier. WSS: wavelength selective switch.}
	\label{fig:Top}
	
\end{figure*}

The abundant capacity and power efficiency of wavelength-division multiplexed (WDM) networks makes them a promising candidate for interconnecting computing nodes in a data center environment and wireless endpoints in a 5G networking scenario. Optical interconnect designs that support the launch of tens of wavelengths per fiber provide for ultrahigh switching capacities, bit-rate transparency, low power density, resource virtualization flexibility, and acceleration in the execution of large-scale distributed applications. Due to their compelling properties, wavelength-routing interconnects, based on the cyclic routing pattern of arrayed waveguide grating (AWG)  as a passive and low-footprint switching device, have received attention for use in both large-scale data centers \cite{rastegarfar2013cross, sato2013large,rastegarfar2017PAM} and the fronthaul segment of radio access networks \cite{zhang2017low}.

While the switching potential of the AWG is limited by its static point-to-point routing pattern (in an $N\times N$ AWG with $N$ available wavelengths, each input port is connected to each output port with a fixed wavelength), AWG-based switch architectures can be made highly flexible by incorporating optical components with complementary switching capabilities \cite{huang2009wavelength, maier2003arrayed,rastegarfar2017PAM}. For instance, star couplers enable nonblocking unicast, multicast, and broadcast traffic delivery directly in the optical domain and can be added to a wavelength-routing architecture to support a rich set of traffic patterns. Multicast traffic involves the simultaneous dissemination of the same information copy to a group of recipients and constitutes a major portion of data center traffic (e.g., MapReduce) \cite{ZhiyangGuo_14,WenKangJia_14, DanLi_12,DanLi_14,HowardWang_13,PaymanSamadi_15,keykhosravi2018multicast,ni2014poxn}. As well, advanced coordinated multipoint (CoMP) transmission techniques in radio access networks call for efficient optical multicasting from a central office (node) to a group of cooperating radio heads \cite{biermann2013backhaul, zhang2017reconfigurable}.

Recently, a hierarchical, wavelength-routing switch architecture with distributed broadcast domains has been proposed for scalable and flexible optical switching in data centers \cite{rastegarfar2017PAM}. As depicted in Fig.~\ref{fig:Top}, this design interconnects a maximum of $N\times(K-1)$ nodes using an $N\times N$ AWG and $N$ $K \times K$ star couplers. The architecture in Fig.~\ref{fig:Top} allows the same set of wavelengths to be used in each broadcast domain (hence, named the wavelength-reuse architecture) in order to overcome the scalability issues due to the limited coupler port count and transceiver tuning range. Ideally, the nodes attached to different couplers should be able to freely communicate with each other. However, two different factors undermine the scalability of the wavelength-reuse architecture. First, with the number of wavelengths and AWG ports being equal, a given pair of couplers  can only use one wavelength to communicate with each other. Second, the physical-layer impairments that accumulate over multiple routing stages can severely limit the achievable capacity per such a wavelength \cite{rastegarfar2018wavelength}. In order to overcome these segmentation mechanisms, it is crucial to improve the interdomain (i.e., intercoupler) switching capacity as much as possible through innovations at both hardware and scheduling levels. 

In this paper, we aim at resolving the interdomain switching bottlenecks of the wavelength-reuse switch architecture by using multiple AWG free spectral ranges (FSRs). In an $N \times N$ AWG with an FSR count of $F$, each input port can be connected to each output port via $F$ distinct wavelengths \cite{bock2005wdm, bock2005hybrid, xu2012large}. Assuming that the range of available wavelengths is fixed, employing a larger $F$ results in a smaller port count. Hence, there is trade-off between the number of supported nodes and the blocking probability (BP) of the design. We develop a mathematical model to approximate this BP and support our analysis with Monte Carlo simulations. Based on a fair, multi-FSR scheduling algorithm, we quantify the impact of $F$ on the blocking performance. We further evaluate the dependence of the physical layer performance measures, i.e., bit error rate (BER) and effective bit rate per node, on the FSR periodicity in the system, considering pulse amplitude modulation (PAM) in a short-reach transmission scenario. Our results illustrate that the switch performance saturates at $F=4$ and that 4-PAM is the most suitable modulation order for multi-FSR interdomain communications. 

The remainder of this paper is organized as follows. In Section~\ref{s2}, we introduce our multi-FSR scheduling algorithm, enabling a fair distribution of parallel wavelength resources among broadcast domains. In Section~\ref{s4}, we derive an analytical approximation of the BPs in the distributed broadcast architecture. In Section~\ref{sim_res}, we study the impact of FSR periodicity via  Monte Carlo simulations and compare the simulation results with the analytical results obtained in Section~\ref{s4}. Section~\ref{s6} is devoted to studying the impact of the physical layer impairments in a multi-FSR switching scenario. We report BER and normalized throughput values, considering various modulation orders, FSR counts, and adaptive coding. 
Finally, Section~\ref{sec:conclution} summarizes and concludes the paper.

\section{A Multi-FSR Scheduling Algorithm  }\label{s2}

In this section, we propose a scheduling algorithm to enable the utilization of multiple FSRs in the architecture of Fig. \ref{fig:Top}. We assume the AWG port count is $N$ and that the number of available wavelengths is $N_W$. Hence, the number of available FSRs is $F=N_W/N$. The value of $F$ indicates the number of wavelengths that connect an arbitrary AWG input port to each output port. Specifically, the $i$th AWG input port  can be connected to the $j$th AWG output port via wavelengths 
\begin{equation}
fN-\mathrm{mod}\left(1-i-j , N\right), \quad f=1, \dots, F.
\end{equation}
Table~\ref{T1} illustrates the input--output connection map of a $4\times 4$ AWG with $F=4$ ($N_W=16$) as an example. 

\begin{table}[!t]
	\centering
	\caption{Routing pattern of a $4\times 4$ AWG with $F=4$. Each input port is connected to  each output port via  $4$ wavelengths.}
	\label{T1}
	\begin{center}
		\begin{tabular}{|c||*{5}{c|}}\hline
			\diaghead{\theadfont Inputaaa}{In}{Out}
			&$1$&$2$&$3$&$4$\\\hline\hline
			$1$&\shortstack{$\lambda_1,\lambda_5,$\\$\lambda_9,\lambda_{13}$}&\shortstack{$\lambda_2,\lambda_6,$\\$\lambda_{10},\lambda_{14}$}&\shortstack{$\lambda_3,\lambda_7,$\\$\lambda_{11},\lambda_{15}$}&\shortstack{$\lambda_4,\lambda_8,$\\$\lambda_{12},\lambda_{16}$}\\\hline
			$2$&\shortstack{$\lambda_2,\lambda_6,$\\$\lambda_{10},\lambda_{14}$}&\shortstack{$\lambda_3,\lambda_7,$\\$\lambda_{11},\lambda_{15}$}&\shortstack{$\lambda_4,\lambda_8,$\\$\lambda_{12},\lambda_{16}$}&\shortstack{$\lambda_1,\lambda_5,$\\$\lambda_9,\lambda_{13}$}\\\hline
			$3$&\shortstack{$\lambda_3,\lambda_7,$\\$\lambda_{11},\lambda_{15}$}&\shortstack{$\lambda_4,\lambda_8,$\\$\lambda_{12},\lambda_{16}$}&\shortstack{$\lambda_1,\lambda_5,$\\$\lambda_9,\lambda_{13}$}&\shortstack{$\lambda_2,\lambda_6,$\\$\lambda_{10},\lambda_{14}$}\\\hline
			$4$&\shortstack{$\lambda_4,\lambda_8,$\\$\lambda_{12},\lambda_{16},$}&\shortstack{$\lambda_1,\lambda_5,$\\$\lambda_9,\lambda_{13}$}&\shortstack{$\lambda_2,\lambda_6,$\\$\lambda_{10},\lambda_{14}$}&\shortstack{$\lambda_3,\lambda_7,$\\$\lambda_{11},\lambda_{15}$}\\\hline
		\end{tabular}
	\end{center}
\end{table}

Let $\lnk{i}{j}$ denote the link from input port $i$ of the AWG to its output port $j$, where $1\leq i,j\leq N$.   Consider two links $\lnk{i}{j}$ and $\lnk{j}{i}$, where $i\neq j$. Due to the reciprocal property of the AWG, the same set of wavelengths can be used to transmit through $\lnk{i}{j}$ and $\lnk{j}{i}$ . If a wavelength is simultaneously used in both links, collision will occur in couplers $i$ and $j$. Therefore, in order to prevent collisions, the scheduling algorithm should take into account the appropriate allocation of wavelengths. In short, a specific wavelength can only be used   either in $\lnk{j}{i}$ or in $\lnk{i}{j}$ but not in both. As a result, with $F=1$, only one of $\lnk{i}{j}$ and $\lnk{j}{i}$ can be allowed to transmit data. With larger $F$ values, however, the available wavelengths can be split and allocated fairly between the two links.

We define some notations that are used in our scheduling algorithm. Let $\wls{i}{j}$ be the set of all wavelengths that can be used to transmit through $\lnk{i}{j}$. As an example, in the AWG represented in Table~\ref{T1},  $\wls{1}{2}=\{\lambda_2,\lambda_6,\lambda_{10},\lambda_{14}\}$. Moreover,  we define two subsets $\wls{i}{j}^{1}$ and $\wls{i}{j}^{2}$ that have the same cardinality and partition the set $\wls{i}{j}$ into two equivalent sets of wavelength resources, e.g., $\wls{1}{2}^{1}=\{\lambda_2,\lambda_6\}$ and $\wls{1}{2}^{2}=\{\lambda_{10},\lambda_{14}\}$\footnote{Here, we assume that $F$ is an even number.}.

Our multi-FSR scheduling algorithm is based on the generalization of the distributed scheduling algorithm proposed in \cite{rastegarfar2017PAM}. It consists of two phases. First, all the interdomain traffic (i.e., connections whose source and destination nodes reside in different couplers) is scheduled and next the intradomain traffic (i.e., connections whose source and destination nodes reside in the same coupler). The scheduling steps for the first phase are as follows.

\
 
\begin{enumerate}
\item \label{step_inter_1}Begin by considering the traffic requests to all destination nodes in  coupler $d$. For fairness, the starting point $1\leq d\leq N$ is chosen randomly or using a round-robin pointer that is updated in each scheduling cycle.
\item \label{step_inter_2}Choose randomly one of the destination nodes in the $d$th coupler with minimum (and non-zero) number of requests.
 \item\label{step_inter_3} Select randomly a source node in another coupler $s$ requesting that destination.
 \item \label{step_inter_4}If $s>d$, pick one available wavelength in $\wls{s}{d}^{1}$ randomly to schedule the connection and mark this wavelength as unavailable from coupler $d$. Block the request if no available wavelength exists. If $s<d$, use  $\wls{s}{d}^{2}$ instead of $\wls{s}{d}^{1}$.
 \item \label{step_inter_5}Repeat  steps \ref{step_inter_2}, \ref{step_inter_3}, and \ref{step_inter_4} until all requests to coupler $d$ are granted or blocked. Afterwards,  schedule the interdomain requests destined to other couplers than $d$ in the same fashion.
 \item \label{step_inter_6} Restore all the blocked requests. To take advantage of any remaining wavelength resources, perform all the previous steps except Step~\ref{step_inter_4}, which is replaced with Step~{4*} as follows.
 \item[4*)]\label{step_inter_4s} If there exists an available wavelength in  $\wls{s}{d}$, use it to schedule the request and mark it as unavailable in coupler $d$. Otherwise, block it.
\item If all couplers $d$ have been examined, terminate the first scheduling phase and go to the second one. Otherwise, update $d$ and go to Step~\ref{step_inter_2}.
\end{enumerate}

\

With $F=1$,  Step~\ref{step_inter_4} is simply replaced with  Step~{4*} and Step~\ref{step_inter_6}
  is removed\footnote{Note that with $F=1$ the cardinality of $\wls{s}{d}$ is equal to one and it cannot be split into two subsets.}. In this case, the algorithm reduces to the one proposed in \cite[Sec.~II-B]{rastegarfar2017PAM}.
In  Steps~\ref{step_inter_1} to \ref{step_inter_5}, the scheduling is performed by fairly dividing the available wavelengths into  two disjoint sets for transmission from coupler $i$ to coupler $j$ or vice versa. In  Step~\ref{step_inter_6}, to ensure work conserving property (see \cite[Sec.~III-C]{keykhosravi2018multicast}), the connections are scheduled using all of the available wavelengths in each link.  Furthermore, in Step~\ref{step_inter_2} the priority is given to destination nodes with the minimum number of requests to minimize the BP \cite{rastegarfar2017PAM}. 

After scheduling the interdomain traffic, the second phase of the algorithm is carried out to schedule the intradomain traffic. This phase is identical to the algorithm in \cite{rastegarfar2017PAM}; however, we present it here for the sake of completeness. According to Fig. \ref{fig:Top}, with $K \times K$ star couplers, each coupler is directly attached to $K-1$ source and destination nodes. The following intradomain traffic scheduling tasks are carried out in each coupler in parallel.

\

\begin{enumerate}

	\item \label{step_intra_1} Begin the scheduling from a destination node $1\leq o\leq K-1$, where $o$ is chosen randomly or according to a round-robin pointer updated in each scheduling cycle.
	\item\label{step_intra_2} If either the destination node has already been matched or no intradomain traffic is destined to that node, go to Step ~\ref{step_intra_4}. Otherwise, randomly select one of the sources requesting that destination.
	\item \label{step_intra_3}  Schedule the connection via a wavelength that has not been used in the coupler. The wavelength can be picked randomly or based on a first-fit wavelength assignment policy. If all wavelengths are occupied, block the request and terminate the scheduling.
	\item \label{step_intra_4} Repeat Step \ref{step_intra_1} for updating $o$ and Steps \ref{step_intra_2} and \ref{step_intra_3} until all of the destinations have been examined.
\end{enumerate}

\section{Blocking Probability Analysis}\label{s4}
In this section, we present an analytical framework for estimating the (interdomain and intradomain) blocking probabilities  of the distributed multicast architecture in Fig.~\ref{fig:Top}.
Here, the BP is defined as the probability that a single connection request is blocked by the scheduler in each scheduling cycle.
 As we conduct offline scheduling, at a scheduling instance, an input port of the switch has a connection request with probability $\rho$. The parameter $\rho$ can be regarded as average input port utilization or the normalized load of the switch. Furthermore, we let $0\leq\rinter\leq 1$ be the probability that a connection is destined to a  node outside its corresponding broadcast domain. The destination  of  each interdomain (intradomain) connection is chosen uniformly among all of the nonlocal (local) nodes.

To conduct our analytical study, firstly, we derive the BP of a single star coupler (i.e., a strict-sense nonblocking switch) in Section~\ref{Sec:41}. Secondly, in  Sections~\ref{sec:FSR1:analytical} and ~\ref{sec:FSR2:analytical}, we calculate the approximate BP of the interdomain traffic for $F=1$ and $F=2$, respectively.  We do not investigate higher values of $F$ due to the increased complexity of the problem. However, we will show through simulations that for $F \geq4$, the BP of the AWG-based switch can be approximated by that of a single nonblocking switch. Thirdly, in Section~\ref{intra_BR} we derive the intradomain BP. 

\subsection{Blocking probability in a star coupler}\label{Sec:41}
Assume that a switch consisting of a single star coupler simultaneously receives $\kin$ connection requests from distinct input ports. Each request is destined to a port that is chosen randomly, independently, and uniformly among $\kout$ output ports. If multiple requests involve the same output port, one of these requests, chosen randomly and uniformly, is accepted, while all other requests for that output port are blocked. The BP is in this scenario given by the following lemma.

\begin{lemma}\label{lemma:1}\normalfont
Each of the $\kin$ connection requests is blocked with probability
\begin{align}\label{eq:average_block_rate}
\BR(\kin,\kout)=1-\frac{\kout-\expe\rlbo{\niop}}{\kin}.
\end{align}
Here, the random variable $\niop$ is the number of idle output ports, whose average is
\begin{align}\label{eq:idel_output_port}
\expe[\niop]=\kout \rlp{1-\frac{1}{\kout}}^{\kin}.
\end{align}

\end{lemma}
\begin{proof}
	The number of blocked connections is $\kin-\nbop$, where $\nbop$ denotes the number of busy output ports, that is, $\kout - \niop$. Therefore, the BP can be calculated as $\expe[\kin-(\kout-\niop)]/\kin$, which is equal to \eqref{eq:average_block_rate}. In the following, we calculate $\expe[\niop]$.
	
	 Define random variable $\mathbf{x}_i, 1\leq i\leq\kout$ to be $1$ if the $i$th output port is idle and $0$ otherwise. The number of idle output ports can be expressed as
	\begin{equation}
		\niop=\sum_{i=1}^{\kout}\mathbf{x}_i.
	\end{equation}  
	Therefore, we have
		\begin{align}\label{eq2}
		\expe[\niop]=\sum_{i=1}^{\kout}\expe\left[\mathbf{x}_i\right].
		\end{align} 
	To calculate $\expe\left[\mathbf{x}_i\right]$, one should note that for all $i$, the $i$th output port is idle  with probability 
	\begin{equation}\label{eq1}
	\pr\left(\mathbf{x}_i=1\right)=\rlp{1-\frac{1}{\kout}}^{\kin}.
	\end{equation}
	Hence, from \eqref{eq2} and \eqref{eq1}, we have
	\begin{align}
	\expe[\niop]&=\kout \rlp{1-\frac{1}{\kout}}^{\kin}.
	\end{align}
\end{proof}
 We note that Lemma~\ref{lemma:1} can be cast into an occupancy problem, where $\kin$ balls (connection requests) are tossed to $\kout$ bins (output ports) (see for example \cite[Ch.~2]{feller1968introduction}).
	

 
Although Lemma~\ref{lemma:1} is formulated to calculate the BP in a star coupler, it can also be used in other nonblocking scenarios. 
We will use  this lemma to derive the BPs of the switch in Fig.~\ref{fig:Top} in the subsequent sections.

\subsection{Interdomain blocking probability with $F=1$}\label{sec:FSR1:analytical}
We consider the case of interdomain routing in the switch of Fig.~\ref{fig:Top} with $F=1$. We denote with $\setx{i}{j}$  the set of all connections whose source and destination lie in couplers $i$ and $j$, respectively. For tractable mathematical analysis, we consider a simplified version of the scheduling algorithm developed in the previous section. With this simplified scheduler, the interdomain connections are scheduled in three steps. First, one connection is chosen at random among each set $\setx{i}{j},\ 1\leq i,j\leq N$, and the rest of the connections are blocked. 
Second, if there exist a connection request from coupler $i$ to coupler $j$ and a connection request from coupler $j$ to $i$, the scheduler randomly chooses one of them and blocks the other\footnote{Note that setting up both connections leads to contention. See Sec.~\ref{s2}.}. Then, each connection is  assigned a wavelength based on the AWG routing pattern. Third, among all non-blocked connections destined to each receiver node, only one  is chosen randomly and the others are blocked (to resolve output port contention). Each receiver is then tuned to the wavelength of its corresponding connection.

Please note that in the original scheduling algorithm (in Sec.~\ref{s2}), the interdomain traffic is scheduled from the output side (i.e, per destination coupler). However, in this analysis we consider a three-stage scheduling, moving from the source to the AWG and then to the destination. Besides, we relax the work-conserving requirement of the scheduler. Let $\bp{1}$, $\bp{2}$, and $\bp{3}$ be the  BP due to the first, second, and third scheduling steps, respectively. The total BP can be calculated as
\begin{equation}
b_{\mathrm{inter}}=1 - (1-\bp{1})(1-\bp{2})(1-\bp{3})\label{btotal_FSR1}.
\end{equation}

In each step of our calculations, we replace all random variables with their mean value to simplify the analysis. As a result, we  assume that $(K-1)\rho$ connections are present in the input ports of each coupler, out of which  
\begin{equation}\label{m1}
\mi=\rinter( K-1) \rho
\end{equation}  
are interdomain connections. 
It is easy to notice that $\bp{1}$ is actually equal to the BP of a coupler with $\mi$ input connections and a destination set of cardinality $N-1$.
Using Lemma~\ref{lemma:1}, we have 
\begin{equation}\label{eq:b1}
\bp{1}\approx\BR(\mi,N-1).
\end{equation}
To calculate $\bp{2}$, we assume  that 
\begin{equation}\label{m2}
\mj=\mi(1-\bp{1})
\end{equation}
connection requests are present on each input port of the AWG.
In the second scheduling step, a connection request from coupler $i$ to coupler $j$ is blocked with probability $1/2$ should there exist a connection request from coupler $j$ to coupler $i$, which happens with probability $\mj/(N-1)$.  Therefore,
\begin{equation}\label{eq:b2}
\bp{2}\approx\frac{\mj}{2(N-1)}. 
\end{equation}
Finally, to calculate  $\bp{3}$, we assume that a total number of $\mk=N\mj(1-\bp{2})$ connections should be scheduled during the third step of the  algorithm. The  destination set of the switch has the cardinality of $N(K-1)$. Here, again the problem can be solved via  Lemma~\ref{lemma:1}. Note that the blocking properties of the switch architecture have already been taken into account in Steps 1 and 2. We obtain
\begin{equation}\label{eq:b3}
\bp{3}\approx\BR(\mk,N(K-1)).
\end{equation}
In \eqref{eq:b3}, for simplification, we neglect  that an interdomain connection is not allowed to be destined to its source domain. This concludes the estimation of the interdomain BP under $F=1$.

\subsection{Interdomain blocking probability with $F=2$}\label{sec:FSR2:analytical}
To calculate the BP with $F=2$, we consider another simplified scheduler that performs the scheduling tasks in Sec.~\ref{sec:FSR1:analytical} twice. 	After  Step~3, the algorithm attempts to schedule the rest of the connections by repeating the first three steps
 as follows.

\

\begin{enumerate}
	\item Randomly choose one connection request from each set $\setx{i}{j}$. We use $\setxt{i}{j}$ to denote the set of  remaining connection requests.
	\item Assign a wavelength   to each of the chosen connections based on their destination. In this step, no connection is blocked as   two wavelengths are available for setting up connections between coupler $i$ and coupler $j$  ($1\leq i,j\leq N$).
	\item For each receiver, one connection is chosen among all the ones destined to it and the rest are blocked.

	\item One connection is randomly selected from each set $\setxt{i}{j}$ and the rest are blocked.
	\item If available, a wavelength is  assigned to each connection according to its destination. Otherwise, the connection is blocked.
	\item For each free receiver node, one connection is selected out of all connections destined to it, and the rest are blocked. Moreover, the connections destined to busy receivers are blocked.   
	
\end{enumerate}

\

 We represent the BP at step $\ell$ by $\bp{\ell}, {1\leq\ell\leq6}$. 
For the first step, $\bp{1}$ can be approximated as in \eqref{eq:b1}. We have $\bp{2}=0$. Similarly to \eqref{eq:b3}, \mbox{$\bp{3}\approx\BR(\mk,N(K-1))$}. The average of the cardinality of set $\setxt{i}{j}$ is $\bp{1}\mi$, where $\mi$ is defined in \eqref{m1}. Therefore, similarly as in \eqref{eq:b1}, 
%
$
\bp{4}\approx\BR(\bp{1}\mi,N-1).
$
%
To approximate $\bp{5}$, we only consider one (the most probable) event, where both wavelengths  for transmission between couplers $i$ and $j$ have been used in the second step, one for transmission from coupler $i$ to $j$ and the other from $j$ to $i$. As discussed in Sec.~\ref{sec:FSR1:analytical}, this probability can be approximated as
%
$
\bp{5}\approx{\mj}/({N-1}),
$
%
where $\mj$ is defined in \eqref{m2}.
 In the sixth step, a connection is blocked either if it is destined to a busy receiver or if it is in contention with other connections. 
The probability of the former event can be approximated by $\bp{6}^{(1)}\approx\ml/(N(K-1))$, where $\ml=N\mi(1-\bp{1})(1-\bp{3})$ is the average number of busy receivers. The probability of the latter event can be calculated similarly as in \eqref{eq:b3} and is 
%
$
\bp{6}^{(2)}\approx\BR(\mm,N(K-1)-\ml),
$
%
where $\mm=N\bp{1}\mi(1-\bp{4})(1-\bp{5})(1-\bp{6}^{(1)})$ is the average number of  connections in Step~$6$ that are destined to free receivers. 
Knowing the BP of each step, the total interdomain BP can be calculated as
\begin{align}
	b_{\mathrm{inter}}&=1 - (1-\bp{1})(1-\bp{2})(1-\bp{3})\nonumber\\
	&\ \ \ \cdot(1-\bp{4})(1-\bp{5})(1-\bp{6}^{(1)})(1-\bp{6}^{(2)}).\label{btotal_FSR2}
\end{align}

\subsection{Intradomain blocking probability  }\label{intra_BR}
After approximating the BP of the interdomain traffic, one can invoke Lemma~\ref{lemma:1} to evaluate that of the intradomain traffic. The average number of busy receivers per coupler, after scheduling the interdomain traffic, can be approximated by
\begin{equation}\label{eq:num:busy}
\nb\approx\mi(1-b_{\mathrm{inter}}).
\end{equation}
The average number of free receivers is $\nf= K-1-\nb$. The intradomain BP consists of two factors: \textit{i)} the probability that a connection is destined to a busy receiver, denoted by $\bpl{1}$ and \textit{ii)} the probability of contention among the connections destined to a free receiver, denoted by $\bpl{2}$. We have that $\bpl{1}\approx\nb/(K-1)$. Moreover, by Lemma~\ref{lemma:1} we have
\begin{equation}\label{eq:bl2}
\bpl{2}\approx\BR\left((1-\rinter) (1-\bpl{1}) (K-1) \rho,\nf\right).
\end{equation}
 Thus, the interdomain BP can be calculated as
\begin{align}
 b_{\mathrm{intra}}&=1 - (1-\bpl{1})(1-\bpl{2})\label{btotal_intra}.
\end{align}

\section{Numerical validation}\label{sim_res}

In this section, we evaluate the blocking probabilities of the distributed multicast architecture via Monte Carlo simulations and compare them with the analytical approximations derived in the previous section for $F\in\{1,2,4,8\}$.   It is assumed that the number of available wavelengths is fixed and equals $N_W=64$. As a result, the AWG port count  scales as $N=64/F$. Each reported value corresponds to the average over $10,000$ simulation runs. Throughout the paper, we set $\rinter=0.25$.
 
 \begin{figure}[!t]
 	\centering
 	\includegraphics{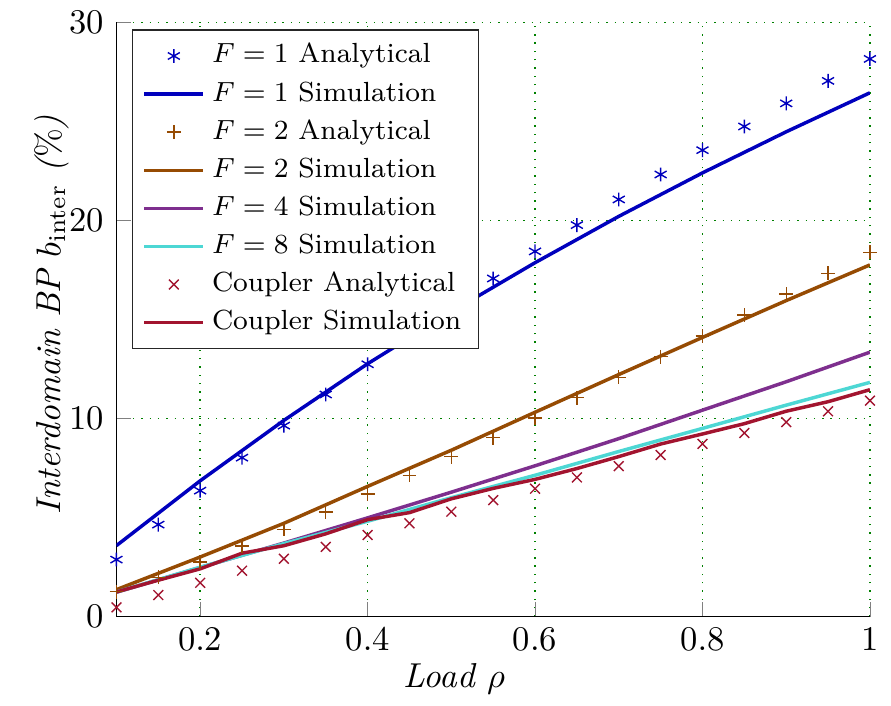}%
 	\caption{ \footnotesize Interdomain BP $b_{\mathrm{inter}}$ of the switch in Fig.~\ref{fig:Top} for $F\in\{1 , 2 , 4 , 8 \}$. Analytical approximations for $F=1$ and $F=2$ are plotted based on \eqref{btotal_FSR1} and \eqref{btotal_FSR2}, respectively.	The BP for a single coupler with $64$ ports,  evaluated via simulations and Lemma~\ref{lemma:1}, is also included for comparison. 
 		 }
 	\label{Fig_analytical}
 \end{figure}

 Figure~\ref{Fig_analytical} represents the interdomain BP versus load.  
Along with the simulation results, the analytical BP values for $F=1$ and $F=2$ are also plotted in Fig.~\ref{Fig_analytical}.
Furthermore, we  include the BP of a single $64\times 64$ coupler. This serves as an approximation for the BP of a switch with $N=2$  couplers ($F=32$), where the  blocking due to the AWG is ignored.
 The BP of the coupler is calculated as 
 \begin{equation}\label{eq:sc}
 b_{\mathrm{coupler}}=\BR\left(\mi,K\right).
 \end{equation}

 As depicted in Fig.~\ref{Fig_analytical}, the analytical results are in close agreement with the simulation results. The small differences are mainly due to two simplifications that we made in our analysis. First, instead of considering the distribution of the random variables,  we considered their expected value, in each step. Second, we analyzed a simplified scheduling algorithm, while the one presented in Sec.~\ref{s2} was simulated.

An interdomain connection is blocked for two reasons, i.e., wavelength shortage in the AWG or output port contention.  The former is more significant for small values of $F$. With an increase in $F$, more wavelengths become available to connect AWG input and output ports. The BP due to wavelength shortage can be significantly reduced with a proper choice of FSR count. According to Fig.~\ref{Fig_analytical}, increasing $F$ past 4 has diminishing returns. 
 It can also be seen that the BP under $F\geq4$ can be approximated by that of a single coupler, since in this scenario blocking becomes contention-dominant.

\begin{figure}[!t]
	\centering
	\includegraphics{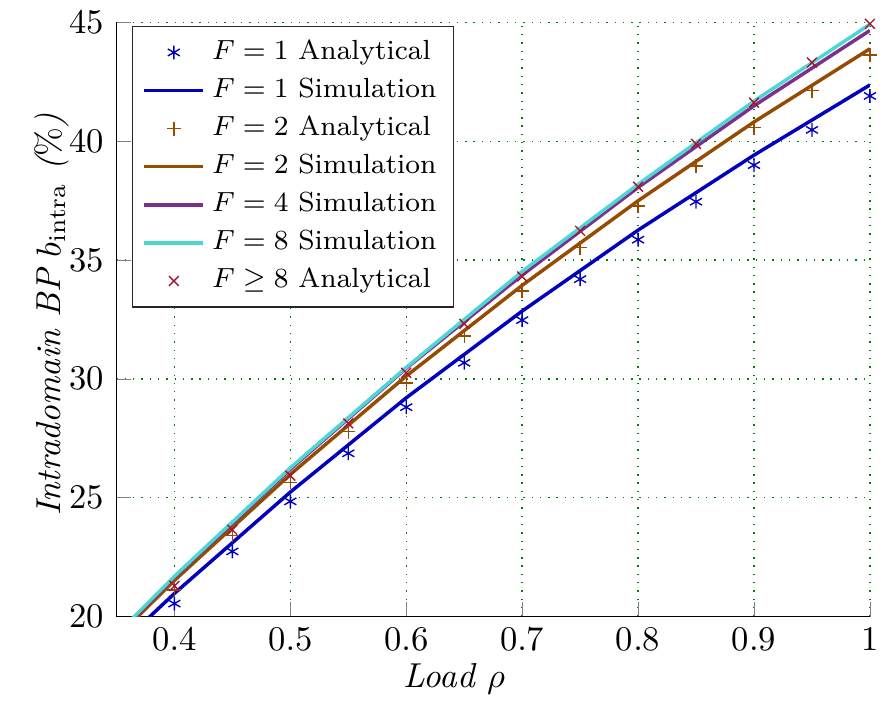}%
	\caption{ \footnotesize Intradomain BP $b_{\mathrm{intra}}$ of the switch in Fig.~\ref{fig:Top} for $F\in\{1 , 2 , 4 , 8 \}$. Analytical approximations for $F=1$, $F=2$ as well as for $F\geq 8$ are also plotted.}
	\label{Fig_analytical_local}
\end{figure}

Figure~\ref{Fig_analytical_local} depicts the simulation results in terms of intradomain BP for  $F = 1 , 2 , 4 ,  8$.  Analytical approximations based on Sec.~\ref{intra_BR} are also provided for $F=1$ and $F=2$. Moreover, for large values of $F$ (i.e., $F\geq8$), we include the intradomain BP based on \eqref{eq:num:busy}--\eqref{btotal_intra}, where  \eqref{eq:sc} was used to calculated $b_{\mathrm{inter}}$ in \eqref{eq:num:busy}. As in the case of interdomain traffic, a good agreement exists between the simulation and analytical results. As depicted in Fig.~\ref{Fig_analytical}, an increase in $F$ results in lower interdomain blocking probabilities; hence, a larger portion of receivers become occupied by the interdomain traffic. This explains why the intradomain BP degrades with an increase in $F$. Besides, for a given load, the intradomain traffic suffers a higher BP compared with the interdomain traffic. This is primarily due to the fact that the multi-FSR scheduler prioritizes the interdomain connections by trying to allocate them resources first.

\section{Cross-Layer Performance Analysis}\label{s6}

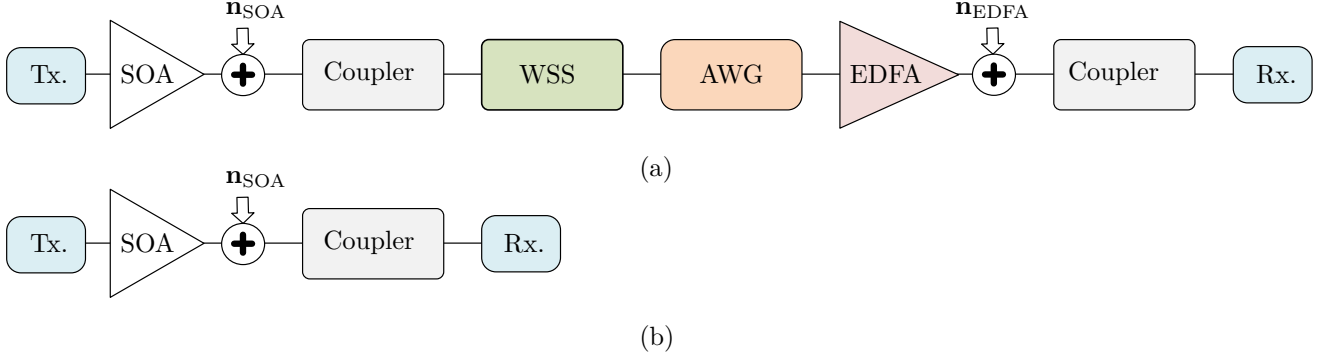
\begin{figure*}[!t]
	\centering
	\begin{tikzpicture}
	\node (ul) at (0,0) {};
	\node (cpp) [right=0cm of ul]{ 	\includegraphics[scale=.6]{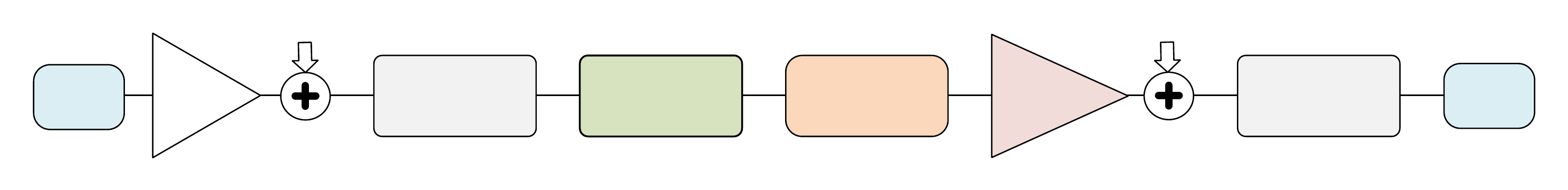}};
	\node (cp)[ right=4.6 of ul]{Coupler};
	\node (cp)[ right=7.2 of ul]{WSS};
	\node (cp)[ right=9.6 of ul]{AWG};
	\node (cp)[ right=1.9 of ul]{SOA};
	\node (cp)[ right=.7 of ul]{Tx.};
	\node (cp)[ right=11.6 of ul]{EDFA};
	\node (cp)[ right=14.5 of ul]{Coupler};
	\node (cp)[ right=17 of ul]{Rx.};
	\node (cp)[ above right=.5cm and 13 of ul]{$\mathbf{n}_{\text{EDFA}}$};
	\node (cp)[ above right=.5cm and 3.3 of ul]{$\mathbf{n}_{\text{SOA}}$};	
	\node (cp)[ above right=-1.7cm and 8.8 of ul]{(a)};
	\node (ul) [ below=2cm of ul]{};
	\node (cpp) [right=0cm of ul]{ 	\includegraphics[scale=.6]{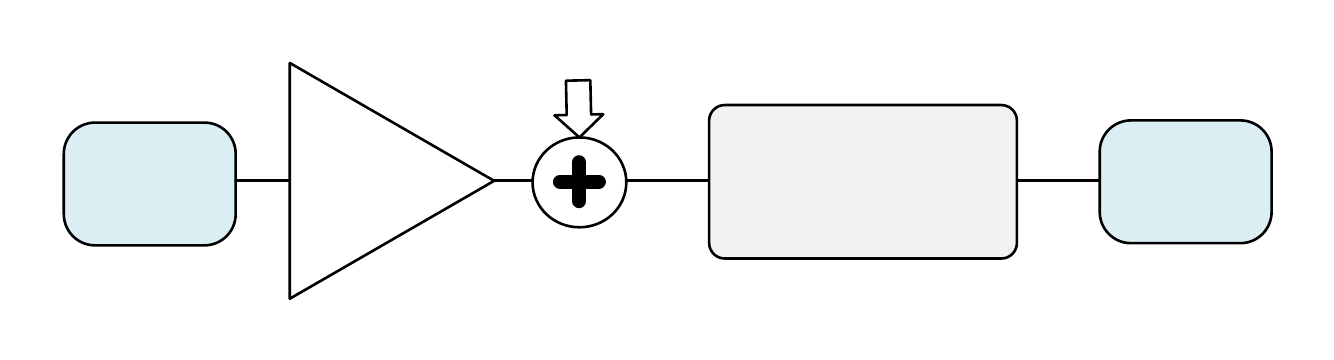}};
	\node (cp)[ right=4.6 of ul]{Coupler};
	\node (cp)[ right=1.9 of ul]{SOA};
	\node (cp)[ right=.7 of ul]{Tx.};
	\node (cp)[ right=7 of ul]{Rx.};
	\node (cp)[ above right=.5cm and 3.3 of ul]{$\mathbf{n}_{\text{SOA}}$};
	\node (cp)[ above right=-1.7cm and 8.8 of ul]{(b)};
	\end{tikzpicture}
	
	\caption{\footnotesize Transmission path  of (a) interdomain and (b) intradomain traffic for the switch in Fig.~\ref{fig:Top}. }
	\label{fig:Pass}
	
\end{figure*}

A signal traversing the different routing stages of the distributed multicast architecture is affected by multiple impairments, namely thermal noise, laser relative intensity noise, shot noise, amplified spontaneous emission (ASE) noise, in-band AWG crosstalk, and out-of-band crosstalk. Therefore, the transmitted symbols are detected at the receiver with errors. Fig.~\ref{fig:Pass} illustrates the signal path from the transmitter to its destination for  interdomain  and  intradomain  traffic. The ASE noises of the amplifiers are shown at the output of the amplification stages. A comprehensive physical-layer model for pulse amplitude modulation (PAM) has been developed in \cite{rastegarfar2017PAM}. We  use the same analytical setup with the same physical-layer parameters (see \cite[Table~I]{rastegarfar2017PAM}) to model the signal propagation and calculate the BER. As in \cite{rastegarfar2017PAM}, we conduct Monte Carlo simulations to evaluate the effects of crosstalk on the signal. Some of the key simulation parameters are presented in Table~\ref{table:parameters}.

\begin{table}[!t]
	\renewcommand{\arraystretch}{1.3}
	\caption{ Simulation parameters }
	\label{table:parameters}
	\centering
	\begin{tabular}{c  c }
		\hline
		\hline
		Parameter & Value\\
		\hline
		Number of simulation runs  & $10,000$\\
		FSR count ($F$)&$1, 2, 4, 8$\\
		Coupler port count  ($K$)&$64$\\
		Wavelength count ($N_W$)  &$64$\\
		AWG port count  ($N$)&$64/F$\\
		Symbol rate  & $28$ Gbaud\\
		$\rinter$  & $0.25$ \\
		\hline
		\hline
	\end{tabular}
\end{table}

\begin{figure*}[!t]
	\centering
	\subfloat[]{\includegraphics{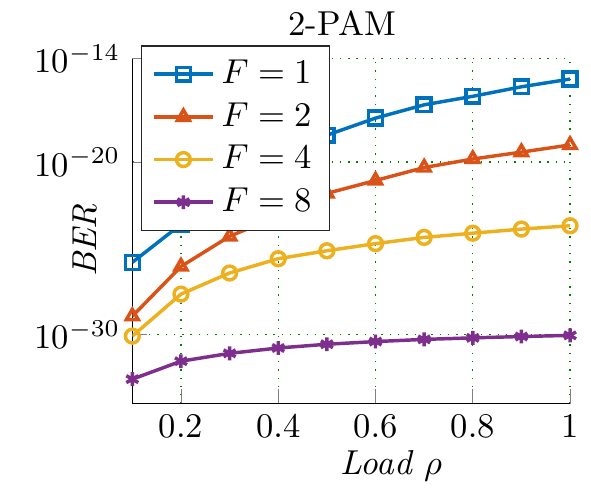}%
		\label{BER_2PAM}}
	\hfil
	\subfloat[]{\includegraphics{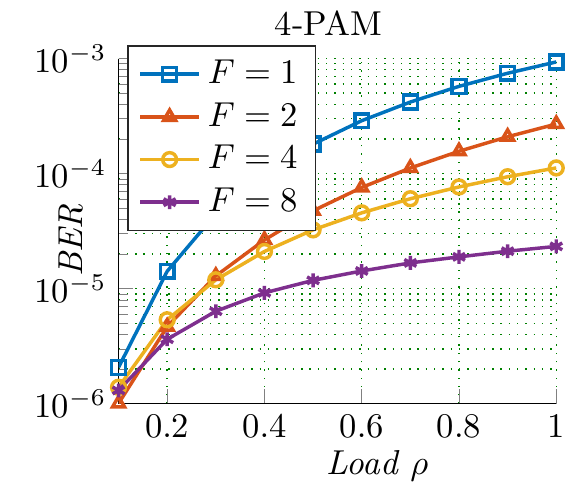}%
		\label{BER_4PAM}}
	\hfil
	\subfloat[]{\includegraphics{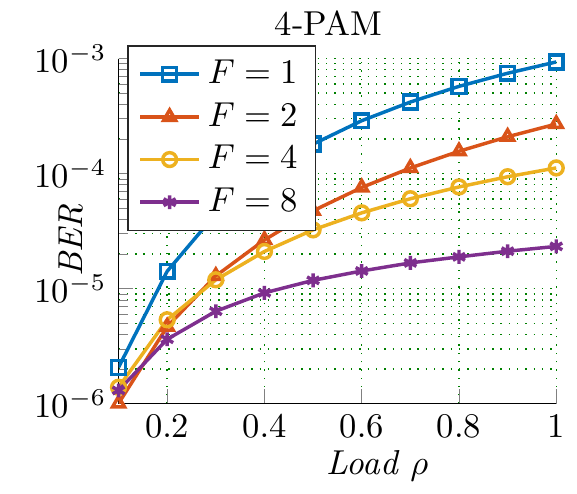}%
		\label{BER_8PAM}}
	\caption{ \footnotesize BER for (a) 2-PAM, (b) 4-PAM, and (c) 8-PAM  for $F = 1 , 2 , 4 , 8$. 
	}
	\label{BER}
\end{figure*}

Figure~\ref{BER} illustrates the (overall) BER versus load for three modulation orders and  $F = 1 , 2 , 4 , 8$.  We can distinguish three trends in Fig.~\ref{BER}. First, an increase in $F$ results in a decrease in BER. This is because for smaller values of $F$, a larger AWG port count is considered (refer to Table~\ref{table:parameters}), and therefore, the interdomain traffic is affected by stronger in-band crosstalk. Second, the BER monotonically increases with load, which is due to the intensified crosstalk originating from more co-propagating channels. Third, increasing the order of the modulation increases the BER. With $8$-PAM, the BER can exceed $10^{-2}$ while with $2$-PAM (i.e., \mbox{on--off}  keying), the transmission is virtually error-free.

To investigate the impact of physical-layer impairments on  the transmission rate for different modulation orders, we deploy a forward error correction (FEC) code with rate adaptation.  We use a Reed--Solomon code with block length of $255$ bytes \cite{ebel1995performance}, i.e., RS$(255,k)$, where $k$ is chosen by the switch controller after calculating the pre-FEC BER such that the  post-FEC BER becomes less than $10^{-12}$. The larger the pre-FEC BER, the smaller the value of $k$, and the lower the effective bit rate per transmitter. We assume that if the pre-FEC BER is larger than $3\cdot10^{-2}$, the signal cannot be retrieved at the receiver.  We focus on the interdomain traffic as it is more susceptible to the impairments than the intradomain traffic. 

\begin{figure*}[!t]
	\centering
	
	\subfloat[]{\includegraphics{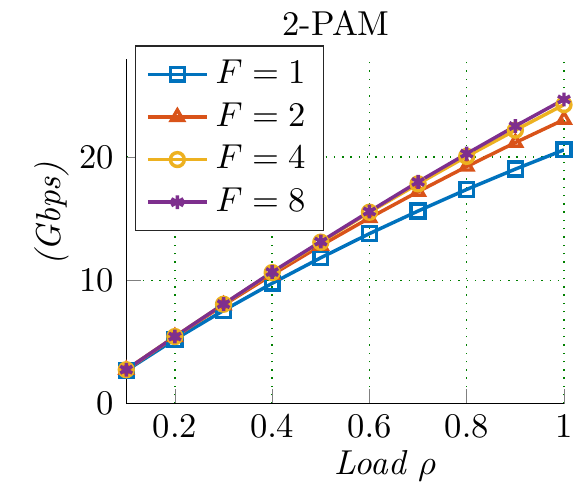}%
		\label{NIT_2PAM}}
	\hfil
	\subfloat[]{\includegraphics{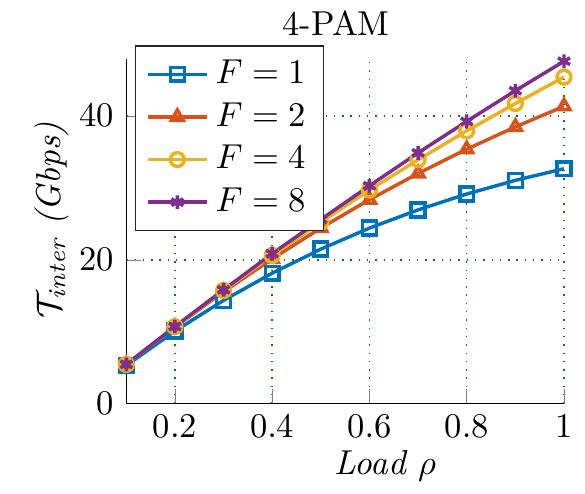}%
		\label{NIT_4PAM}}
	\hfil
	\subfloat[]{\includegraphics{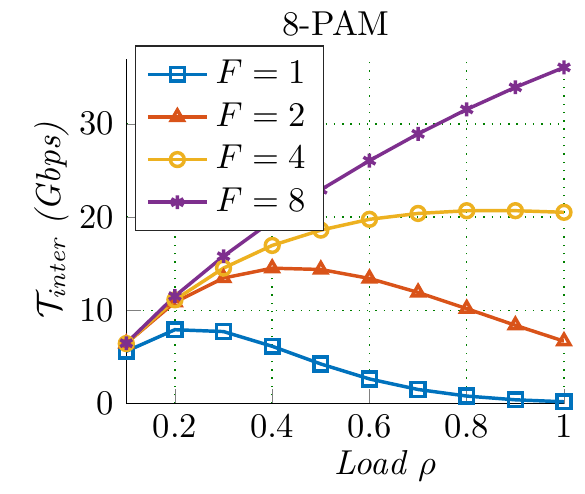}%
		\label{NIT_8PAM}}
	\caption{ \footnotesize $\nidt$ (see \eqref{eq_tau}) for (a) 2-PAM, (b) 4-PAM, and (c) 8-PAM for $F = 1 , 2 , 4 , 8$. 
	}
	\label{NIT}
\end{figure*}
Figure~\ref{NIT} depicts $\nidt$, that is the  average interdomain throughput per node  normalized by $\rinter$, versus load for $F\in\{ 1 , 2 , 4 , 8 \}$. 
 Specifically, the vertical axis in Fig.~\ref{NIT} represents
  \begin{equation}\label{eq_tau}
  \nidt=\frac{\text{Average total interdomain traffic }}{  N(K-1)  \times \rinter}.
  \end{equation}
For each modulation order, $\nidt$ grows with $F$. As is evident from Fig.~\ref{BER}, the larger the value of $F$, the lower the pre-FEC BER; hence, a larger throughput. Increasing the load has two opposing effects on $\nidt$. First, with higher loads, more connections are set up, translating to higher throughput. Second, the BER increases with load (see Fig.~\ref{BER}), and consequently the code rate and throughput decrease.   With $2$-PAM and $4$-PAM, the first effect dominates and the throughput constantly increases with load. However, with $8$-PAM, the second effect becomes dominant for $F=1,2$ under large enough loads.

Using higher-order modulations  \textit{i)} increases the number of transmitted bits per symbol and consequently the throughput, and \textit{ii)} increases the pre-FEC BER, which in turn decreases the code rate and throughput. As depicted in Fig.~\ref{NIT}, by moving from $2$-PAM to $4$-PAM, the  throughput increases as the first effect dominates. However, comparing Fig.~\ref{NIT}\,\subref{NIT_4PAM} and Fig.~\ref{NIT}\,\subref{NIT_8PAM}, one can note that for a given load, the throughput decreases by moving from $4$-PAM to $8$-PAM. Therefore, $4$ is the best modulation order to be used by the interdomain traffic in the multi-FSR transmission scenario.



\section{ Conclusion}\label{sec:conclution}

To address the switching bottlenecks imposed by the fixed AWG routing pattern, we investigated the impact of AWG FSR periodicity on the performance of a distributed multicast switch architecture. We developed a general-purpose, analytical framework to estimate the BP in a multistage setting. In addition, we conducted Monte Carlo simulations to study the performance under different values of the FSR count. Considering the parameters of our study, the interdomain BP could be significantly improved by increasing the FSR count from one to four, with larger values resulting in diminishing returns and further reducing the number of supported nodes.

From a physical-layer standpoint, an increase in the FSR count leads to a decrease in BER and a larger effective bit rate per connection. In our cross-layer simulations, $4$-PAM led to the highest normalized interdomain throughput for all of the considered FSR counts. 
 Resolving the broadcast domain partitioning problem by using multiple FSRs, our future work involves the development of new scheduling algorithms giving equal priority to intradomain and interdomain traffic.

\section*{Acknowledgment}
This work was supported by the Swedish Research Council under grant no. 2014-6230, the NSF Center for Integrated Access Networks (CIAN) under grant no. EEC-0812072, and the Natural Sciences and Engineering Research Council of Canada (NSERC).
The simulations were carried out on the resources provided by the Swedish National Infrastructure for Computing (SNIC) at C3SE.



\begin{thebibliography}{10}
	\providecommand{\url}[1]{#1}
	\csname url@samestyle\endcsname
	\providecommand{\newblock}{\relax}
	\providecommand{\bibinfo}[2]{#2}
	\providecommand{\BIBentrySTDinterwordspacing}{\spaceskip=0pt\relax}
	\providecommand{\BIBentryALTinterwordstretchfactor}{4}
	\providecommand{\BIBentryALTinterwordspacing}{\spaceskip=\fontdimen2\font plus
		\BIBentryALTinterwordstretchfactor\fontdimen3\font minus
		\fontdimen4\font\relax}
	\providecommand{\BIBforeignlanguage}[2]{{%
			\expandafter\ifx\csname l@#1\endcsname\relax
			\typeout{** WARNING: IEEEtran.bst: No hyphenation pattern has been}%
			\typeout{** loaded for the language `#1'. Using the pattern for}%
			\typeout{** the default language instead.}%
			\else
			\language=\csname l@#1\endcsname
			\fi
			#2}}
	\providecommand{\BIBdecl}{\relax}
	\BIBdecl
	
	\bibitem{agyapong2014design}
	P.~K. Agyapong, M.~Iwamura, D.~Staehle, W.~Kiess, and A.~Benjebbour, ``Design
	considerations for a {5G} network architecture,'' \emph{IEEE Communications
		Magazine}, vol.~52, no.~11, pp. 65--75, Nov. 2014.
	
	\bibitem{wang2015backhauling}
	N.~Wang, E.~Hossain, and V.~K. Bhargava, ``{Backhauling 5G small cells: a radio
		resource management perspective},'' \emph{IEEE Wireless Communications},
	vol.~22, no.~5, pp. 41--49, Oct. 2015.
	
	\bibitem{velasco2017meeting}
	L.~Velasco, A.~Castro, A.~Asensio, M.~Ruiz, G.~Liu, C.~Qin, R.~Proietti, and
	S.~J.~B. Yoo, ``Meeting the requirements to deploy cloud {RAN} over optical
	networks,'' \emph{Journal of Optical Communications and Networking}, vol.~9,
	no.~3, pp. B22--B32, Mar. 2017.
	
	\bibitem{raza2016demonstration}
	M.~R. Raza, M.~Fiorani, A.~Rostami, P.~{\"O}hl{\'e}n, L.~Wosinska, and
	P.~Monti, ``{Demonstration of dynamic resource sharing benefits in an optical
		C-RAN},'' \emph{Journal of Optical Communications and Networking}, vol.~8,
	no.~8, pp. 621--632, Aug. 2016.
	
	\bibitem{gowda2016quasi}
	A.~S. Gowda, L.~G. Kazovsky, K.~Wang, and D.~Larrabeiti, ``Quasi-passive
	optical infrastructure for future 5{G} wireless networks: pros and cons,''
	\emph{Journal of Optical Communications and Networking}, vol.~8, no.~12, pp.
	B111--B123, Dec. 2016.
	
	\bibitem{tomkos2013optical}
	C.~Kachris, K.~Kanonakis, and I.~Tomkos, ``Optical interconnection networks in
	data centers: recent trends and future challenges,'' \emph{IEEE
		Communications Magazine}, vol.~51, no.~9, pp. 39--45, Sep. 2013.
	
	\bibitem{JiajiaChen_15}
	J.~Chen, Y.~Gong, M.~Fiorani, and S.~Aleksic, ``Optical interconnects at the
	top of the rack for energy-efficient data centers,'' \emph{IEEE
		Communications Magazine}, vol.~53, no.~8, pp. 140--148, Aug. 2015.
	
	\bibitem{rastegarfar2017PAM}
	H.~Rastegarfar, L.~Yan, K.~Szczerba, and E.~Agrell, ``{PAM} performance
	analysis in multicast-enabled wavelength-routing data centers,''
	\emph{Journal of Lightwave Technology}, vol.~35, no.~13, pp. 2569--2579, Jul.
	2017.
	
	\bibitem{rastegarfar2013cross}
	H.~Rastegarfar, A.~Leon-Garcia, S.~LaRochelle, and L.~A. Rusch, ``Cross-layer
	performance analysis of recirculation buffers for optical data centers,''
	\emph{Journal of lightwave Technology}, vol.~31, no.~3, pp. 432--445, Feb.
	2013.
	
	\bibitem{sato2013large}
	K.-I. Sato, H.~Hasegawa, T.~Niwa, and T.~Watanabe, ``A large-scale wavelength
	routing optical switch for data center networks,'' \emph{IEEE Communications
		Magazine}, vol.~51, no.~9, pp. 46--52, Sep. 2013.
	
	\bibitem{zhang2017low}
	K.~Zhang, Q.~Zhuge, H.~Xin, H.~He, W.~Hu, and D.~V. Plant, ``{Low-cost WDM
		fronthaul enabled by partitioned asymmetric AWGR with simultaneous flexible
		transceiver assignment and chirp management},'' \emph{IEEE/OSA Journal of
		Optical Communications and Networking}, vol.~9, no.~10, pp. 876--888, Oct.
	2017.
	
	\bibitem{huang2009wavelength}
	Q.~Huang and W.-D. Zhong, ``Wavelength-routed optical multicast packet switch
	with improved performance,'' \emph{Journal of Lightwave Technology}, vol.~27,
	no.~24, pp. 5657--5664, Dec. 2009.
	
	\bibitem{maier2003arrayed}
	M.~Maier, M.~Scheutzow, and M.~Reisslein, ``The arrayed-waveguide grating-based
	single-hop {WDM} network: an architecture for efficient multicasting,''
	\emph{IEEE Journal on Selected Areas in Communications}, vol.~21, no.~9, pp.
	1414--1432, Nov. 2003.
	
	\bibitem{ZhiyangGuo_14}
	Z.~Guo, J.~Duan, and Y.~Yang, ``On-line multicast scheduling with bounded
	congestion in fat-tree data center networks,'' \emph{IEEE Journal on Selected
		Areas in Communications}, vol.~32, no.~1, pp. 102--115, Jan. 2014.
	
	\bibitem{WenKangJia_14}
	W.-K. Jia, ``A scalable multicast source routing architecture for data center
	networks,'' \emph{IEEE Journal on Selected Areas in Communications}, vol.~32,
	no.~1, pp. 116--123, Jan. 2014.
	
	\bibitem{DanLi_12}
	D.~Li, Y.~Li, J.~Wu, S.~Su, and J.~Yu, ``{ESM}: Efficient and scalable data
	center multicast routing,'' \emph{IEEE/ACM Transactions on Networking},
	vol.~20, no.~3, pp. 944--955, Jun. 2012.
	
	\bibitem{DanLi_14}
	D.~Li, M.~Xu, Y.~Liu, X.~Xie, Y.~Cui, J.~Wang, and G.~Chen, ``Reliable
	multicast in data center networks,'' \emph{IEEE Transactions on Computers},
	vol.~63, no.~8, pp. 2011--2024, Aug. 2014.
	
	\bibitem{HowardWang_13}
	H.~Wang, Y.~Xia, K.~Bergman, T.~S. Ng, S.~Sahu, and K.~Sripanidkulchai,
	``Rethinking the physical layer of data center networks of the next decade:
	Using optics to enable efficient \mbox{*-cast} connectivity,'' \emph{ACM
		SIGCOMM Computer Communication Review}, vol.~43, no.~3, pp. 52--58, Jul.
	2013.
	
	\bibitem{PaymanSamadi_15}
	P.~Samadi, V.~Gupta, J.~Xu, H.~Wang, G.~Zussman, and K.~Bergman, ``Optical
	multicast system for data center networks,'' \emph{Optics Express}, vol.~23,
	no.~17, pp. 22\,162--22\,180, Aug. 2015.
	
	\bibitem{keykhosravi2018multicast}
	K.~Keykhosravi, H.~Rastegarfar, and E.~Agrell, ``Multicast scheduling of
	wavelength-tunable, multiqueue optical data center switches,'' \emph{Journal
		of Optical Communications and Networking}, vol.~10, no.~4, pp. 353--364, Apr.
	2018.
	
	\bibitem{ni2014poxn}
	W.~Ni, C.~Huang, Y.~L. Liu, W.~Li, K.-W. Leong, and J.~Wu, ``{POXN}: a new
	passive optical cross-connection network for low-cost power-efficient
	datacenters,'' \emph{Journal of Lightwave Technology}, vol.~32, no.~8, pp.
	1482--1500, Apr. 2014.
	
	\bibitem{biermann2013backhaul}
	T.~Biermann, L.~Scalia, C.~Choi, W.~Kellerer, and H.~Karl, ``How backhaul
	networks influence the feasibility of coordinated multipoint in cellular
	networks,'' \emph{IEEE Communications Magazine}, vol.~51, no.~8, pp.
	168--176, Aug. 2013.
	
	\bibitem{zhang2017reconfigurable}
	J.~Zhang, Y.~Ji, S.~Jia, H.~Li, X.~Yu, and X.~Wang, ``Reconfigurable optical
	mobile fronthaul networks for coordinated multipoint transmission and
	reception in {5G},'' \emph{Journal of Optical Communications and Networking},
	vol.~9, no.~6, pp. 489--497, Jun. 2017.
	
	\bibitem{rastegarfar2018wavelength}
	H.~Rastegarfar, K.~Keykhosravi, E.~Agrell, and N.~Peyghambarian, ``Wavelength
	reuse for scalable multicasting: a cross-layer perspective,'' in
	\emph{Optical Fiber Communication Conference}, Mar. 2018, p. W2A.20.
	
	\bibitem{bock2005wdm}
	C.~Bock and J.~Prat, ``{WDM/TDM PON experiments using the AWG free spectral
		range periodicity to transmit unicast and multicast data},'' \emph{Optics
		Express}, vol.~13, no.~8, pp. 2887--2891, Apr. 2005.
	
	\bibitem{bock2005hybrid}
	C.~Bock, J.~Prat, and S.~D. Walker, ``{Hybrid WDM/TDM PON using the AWG FSR and
		featuring centralized light generation and dynamic bandwidth allocation},''
	\emph{Journal of Lightwave Technology}, vol.~23, no.~12, pp. 3981--3988, Dec.
	2005.
	
	\bibitem{xu2012large}
	Z.~Xu, X.~Cheng, Y.-K. Yeo, X.~Shao, L.~Zhou, and H.~Zhang, ``{Large-scale WDM
		passive optical network based on cyclical AWG},'' \emph{Optics Express},
	vol.~20, no.~13, pp. 13\,939--13\,946, Jun. 2012.
	
	\bibitem{feller1968introduction}
	W.~Feller, \emph{An introduction to probability theory and its applications},
	3rd~ed.\hskip 1em plus 0.5em minus 0.4em\relax Wiley, New York, 1968, vol.~1.
	
	\bibitem{ebel1995performance}
	W.~J. Ebel and W.~H. Tranter, ``The performance of {Reed-Solomon} codes on a
	bursty-noise channel,'' \emph{IEEE Trans.\ Commun.}, vol.~43, no. 2/3/4, pp.
	298--306, Feb.--Apr. 1995.
	
\end{thebibliography}
\end{document}